\PassOptionsToPackage{unicode}{hyperref}
\PassOptionsToPackage{hyphens}{url}
\documentclass[]{article}
\usepackage{amsmath,amssymb,bbm,amsthm}
\usepackage{natbib}
\usepackage{lmodern}
\usepackage{iftex}
\usepackage{enumitem}
\ifPDFTeX
  \usepackage[T1]{fontenc}
  \usepackage[utf8]{inputenc}
  \usepackage{textcomp} % provide euro and other symbols
\else % if luatex or xetex
  \usepackage{unicode-math}
  \defaultfontfeatures{Scale=MatchLowercase}
  \defaultfontfeatures[\rmfamily]{Ligatures=TeX,Scale=1}
\fi
\IfFileExists{upquote.sty}{\usepackage{upquote}}{}
\IfFileExists{microtype.sty}{% use microtype if available
  \usepackage[]{microtype}
  \UseMicrotypeSet[protrusion]{basicmath} % disable protrusion for tt fonts
}{}
\makeatletter
\@ifundefined{KOMAClassName}{% if non-KOMA class
  \IfFileExists{parskip.sty}{%
    \usepackage{parskip}
  }{% else
    \setlength{\parindent}{0pt}
    \setlength{\parskip}{6pt plus 2pt minus 1pt}}
}{% if KOMA class
  \KOMAoptions{parskip=half}}
\makeatother
\usepackage{xcolor}
\IfFileExists{xurl.sty}{\usepackage{xurl}}{} % add URL line breaks if available
\IfFileExists{bookmark.sty}{\usepackage{bookmark}}{\usepackage{hyperref}}
\hypersetup{
  pdfauthor={James Long},
  hidelinks,
  pdfcreator={LaTeX via pandoc}}
\usepackage[margin=1in]{geometry}
\usepackage{graphicx}
\makeatletter
\def\maxwidth{\ifdim\Gin@nat@width>\linewidth\linewidth\else\Gin@nat@width\fi}
\def\maxheight{\ifdim\Gin@nat@height>\textheight\textheight\else\Gin@nat@height\fi}
\makeatother
\usepackage{caption}
\usepackage{subcaption}

\setkeys{Gin}{width=\maxwidth,height=\maxheight,keepaspectratio}
\makeatletter
\def\fps@figure{htbp}
\makeatother
\ifLuaTeX
  \usepackage{selnolig}  % disable illegal ligatures
\fi

\title{Causal Models, Prediction, and Extrapolation in Cell Line Perturbation Experiments}
\author{James P. Long\thanks{Department of Biostatistics, University of Texas MD Anderson Cancer Center}\, \footnote{Corresponding Author: jplong@mdanderson.org} \, \, \, \, \,  Yumeng Yang\thanks{Biomedical Informatics, The University of Texas Health Science Center at Houston}  \, \, \, \, \,  Kim-Anh Do \footnotemark[1]}
\date{\today}

\newcommand{\argmin}[1]{{\underset{#1}{\textnormal{argmin}} \, }}

\newcommand{\bs}[1]{\mathbf{#1}}

\newtheorem{Thm}{\underline{\bf Theorem}}
\newtheorem{Lem}{\underline{\bf Lemma}}

\begin{document}
\maketitle

{
\setcounter{tocdepth}{2}
\tableofcontents
}
\hypertarget{abstract}{%
\section*{Abstract}\label{abstract}}

In cell line perturbation experiments, a collection of cells is perturbed with external agents (e.g. drugs) and responses such as protein expression measured. Due to cost constraints, only a small fraction of all possible perturbations can be tested in vitro. This has led to the development of computational (in silico) models which can predict cellular responses to perturbations. Perturbations with clinically interesting predicted responses can be prioritized for in vitro testing. In this work, we compare causal and non-causal regression models for perturbation response prediction in a Melanoma cancer cell line. The current best performing method on this data set is Cellbox which models how proteins causally effect each other using a system of ordinary differential equations (ODEs). We derive a closed form solution to the Cellbox system of ODEs in the linear case. These analytic results facilitate comparison of Cellbox to non--causal regression approaches. We show that causal models such as Cellbox, while requiring more assumptions, enable extrapolation in ways that non-causal regression models cannot. For example, causal models can predict responses for never before tested drugs. We illustrate these strengths and weaknesses in simulations. In an application to the Melanoma cell line data, we find that regression models outperform the Cellbox causal model. 

\hypertarget{introduction}{%
\section{Introduction}\label{introduction}}

In cell line perturbation experiments, a collection of cells is perturbed with gene knockdowns, overexpression, or pharmaceutical drugs and responses such as cell survival and gene and protein expression are measured. The results of these experiments play an important role in our understanding of cellular biology and in development of treatments for complex diseases such as cancer \citep{zhao2020large,subramanian2017next,korkut2015perturbation}.

There are a huge number of possible perturbations that can be applied to a cell line. For example, in human cell lines there are $\sim 20,000$ genes which could be perturbed (e.g. knocked out). Thus there are $\sim 200$ million perturbations of gene pairs (double knockouts). Further each perturbation may be applied across hundreds of cell lines (e.g. cells of different types of cancer). Thus in practice even very large scale experiments can only test a small set of all possible perturbations.

This limitation has led to the development of in silico perturbation response prediction models \citep{squires2020causal,lotfollahi2019scgen,yuan2021cellbox,korkut2015perturbation}. Models are typically trained on a set of perturbations which are experimentally tested in a laboratory and where cellular responses to the perturbation are known (up to technical replicate variability). These in silico models can then be used to make response predictions for untested perturbations. Predicted responses of biological interest, e.g. a perturbation which is predicted to suppress growth in a tumor cell line, can then be experimentally validated in vitro.

One class of perturbation prediction methods uses the training data to construct a causal model specifying how response variables such as gene or protein expression influence each other \citep{meinshausen2016methods,rothenhausler2019,sachs2005causal,squires2020causal}. For example, a model may assume that gene expressions follow a causal Directed Acyclic Graph (DAG). Perturbations are used to estimate the DAG (called a gene regulatory network (GRN) in this context). The GRN can be used for perturbation prediction because if the direct targets of a perturbation are known (e.g. a knockdown of gene X directly changes gene X), then the downstream effects of this perturbation on other genes can be inferred from the GRN.

%In the field of causal inference, models are typically used for parameter inference and hypothesis testing with validity of these outputs (e.g. coverage probability of confidence intervals) depending on causal assumptions. These causal assumptions, such as no hidden confounding, may be untestable even asymptotically. In contrast, causal models for perturbation experiments are typically evaluated based on prediction performance. Specifically perturbations which were performed in the laboratory (and for which responses are known) are divided into training and test sets. Model parameters are learned on the training set and then prediction performance is evaluated on the test set.

%Here we compare regression and causal models for prediction of perturbation responses. study the ability of causal models to extrapolate when making predictions. 

%Consider a prediction problem where in the training set one feature is always $0$, but in the test set the same feature is never $0$. This represents an extreme form of extrapolation. In general, a regression model will be unable to successfully model the dependence of the response on the feature. For example the ordinary least squares estimator is undefined in this context.

%Here we show that this form of extrapolation is closely linked to the problem of perturbation response prediction. In fact, causal discovery methods are, in theory, capable of 

Here we explore the relationship between regression and causal models in cell line perturbation response prediction. We consider a recently proposed causal method Cellbox which models protein responses to perturbations using a system of ordinary differential equations (ODEs) which specify how proteins causally effect each other across time \citep{yuan2021cellbox}. In recent work, Cellbox outperformed competitor methods, including co-expression models, Bayesian networks, and Neural Networks, in predicting the protein and phenotype responses of melanoma cell line SK-Mel-133 to drug perturbations.

We derive a closed form solution to the Cellbox system of ODEs in the linear case. These analytic results facilitate comparison of Cellbox to a regression approach. Cellbox requires knowledge of the direct targets of treatment and must estimate a protein regulatory network. On these two points, the regression model is simpler and makes fewer assumptions. However Cellbox can extrapolate to predict the effect of previously untested drugs whereas the regression model is not formally defined in this context. We illustrate the strengths and weaknesses of regression and causal models in simulations. These are some of the first results connecting causal modeling, extrapolation, and prediction.

%Next we consider perturbation prediction in the Kemmeren yeast gene knockout data set. These data set have been extensively used to study the performance of causal modelling for domain adaptation and transfer learning. We show that the Kemmeren data structure can be represented in a form similar to the Melanoma data and that prediction in Kemmeren can be represented as an extrapolation challenge. These results provide the first connection between domain adaptation, transfer learning, and causal modeling. Further, these results suggest a novel estimator for response prediction with Kemmeren.

Next we compare the performance of the Cellbox causal model to a regression model on the original motivating Melanoma cell line data. We compare Cellbox to regression with random fold cross validation and leave on drug out cross validation. The latter form of validation favors Cellbox because the regression model has no way of inferring the effect of the drug not used in the training data. Surprisingly, we find that regression has superior/equal performance in both settings. These results suggest that Cellbox may be better suited for application in other settings, such as with richer time domain data, more perturbations, or regulatory networks which more closely follow modeling assumptions. The regression model we propose should be used as a benchmark in future studies with this Melanoma cell line (or similar data sets) since it has the best performance relative to methods tested in the literature thus far.

This work is organized as follows. Section \ref{overview} introduces the Melanoma cell line perturbation data and summarizes how causal models can be used for prediction when extrapolating outside the domain of the training data. Section \ref{cellbox} derives analytic results for comparing the causal model Cellbox with regression approaches for predicting cellular responses to perturbations. Simulations are used to demonstrate the strengths and weakness of different approaches to prediction. Section \ref{mela} contains numerical results comparing Cellbox with regression on the Melanoma cell line perturbation data. Section \ref{discussion} contains conclusion and discussion. All code for reproducing the results in this work is publicly available.\footnote{\url{https://github.com/longjp/causal-pred-drug-code}}

%There has been much recent interest in the statistics literature on the ability of causal models to make predictions in new data generating domains, not represented in training data \citep{meinshausen2018causality,peters2016causal,rojas2018invariant}. Regression models which learn a conditional distribution $p(Y|X)$ in one domain, may perform poorly in a new domain where $X$ is perturbed. In contrast, a causal model which estimates $p(Y|do(X))$ should be more robust. 

% Describe how causal models work: construct GRN and then use to predict knockdown of gene X. one such example is cellbox. causal discovery + prediction

% Describe what we do. relate to statistics literature on generalizability of models

% summary of section of this work

\hypertarget{overview}{%
\section{Overview of Data and Causal Prediction}\label{overview}}

In this work we consider perturbation experiments on a RAFi-resistant melanoma cell line SkMel-133 originally collected in \cite{korkut2015perturbation}. The data structure is depicted in Figure \ref{fig:lodo-struct}. The cell line was treated with 89 drug perturbations (rows). Perturbations are defined by the doses of 12 drugs (blue columns). Drugs were applied as a single agent and in combinations of two drugs. Since at most two drugs were used in any experiment, each row of the blue columns contains either 1 (if perturbation is with single drug) or 2 (if two drug perturbation was performed) non-zero values. In each perturbation experiment, the expression of 82 proteins was measured 24 hours after perturbation using Reverse Phase Protein Arrays \citep{tibes2006reverse}. In addition 5 cell phenotypes were measured, quantifying cell-cycle progression and cell viability (orange columns). For this work, we use the data as supplied by \cite{yuan2021cellbox}.\footnote{Available at \url{https://github.com/sanderlab/CellBox}.}

\begin{figure}
 \centering
  \begin{subfigure}{0.49\textwidth}
   \centering
   \includegraphics[width=0.95\linewidth]{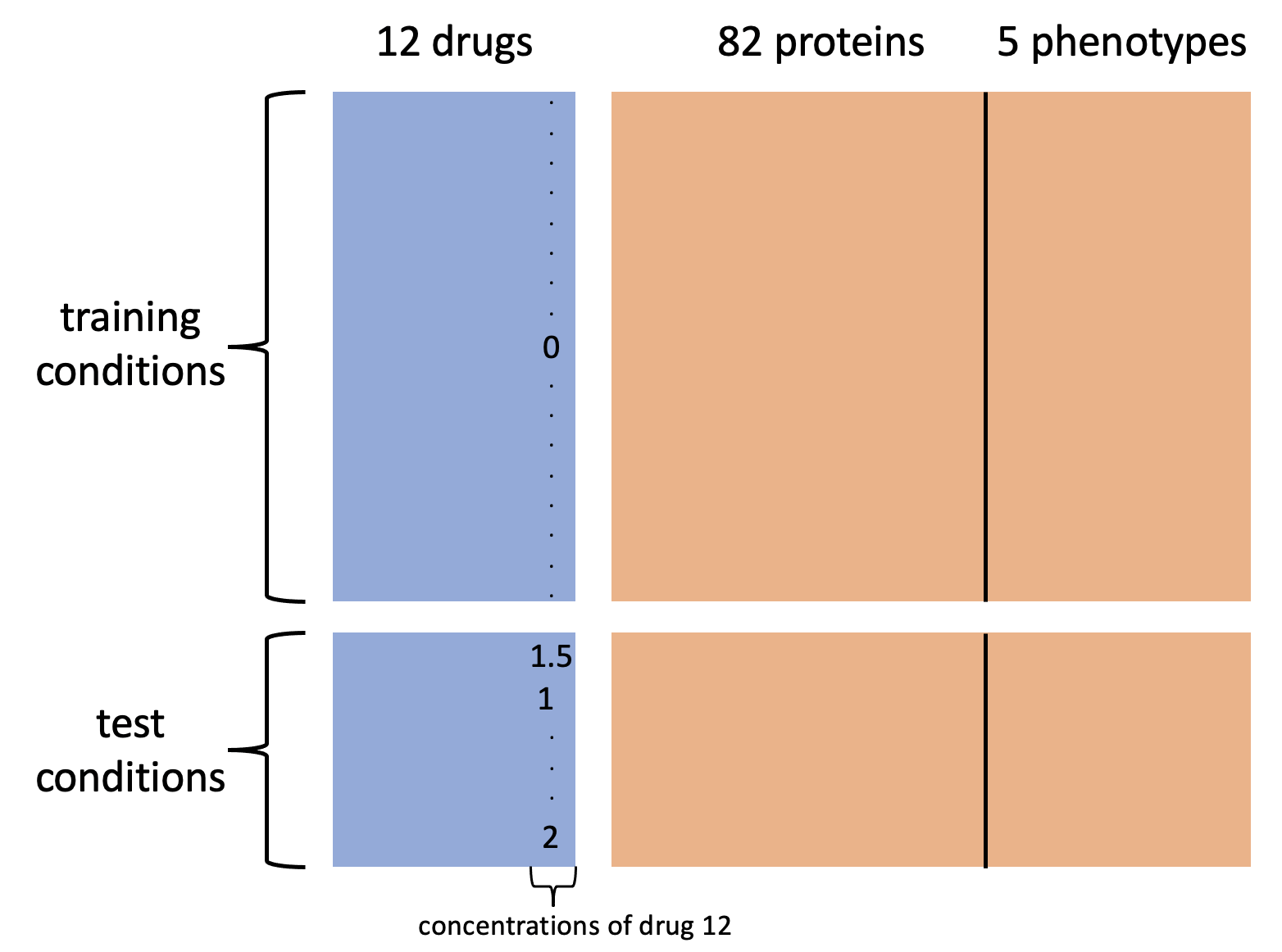}
   \caption{Melanoma Cell Line Data Structure}
   \label{fig:lodo-struct}
   \end{subfigure}%
   \begin{subfigure}{0.49\textwidth}
    \centering
    \includegraphics[width=0.95\linewidth]{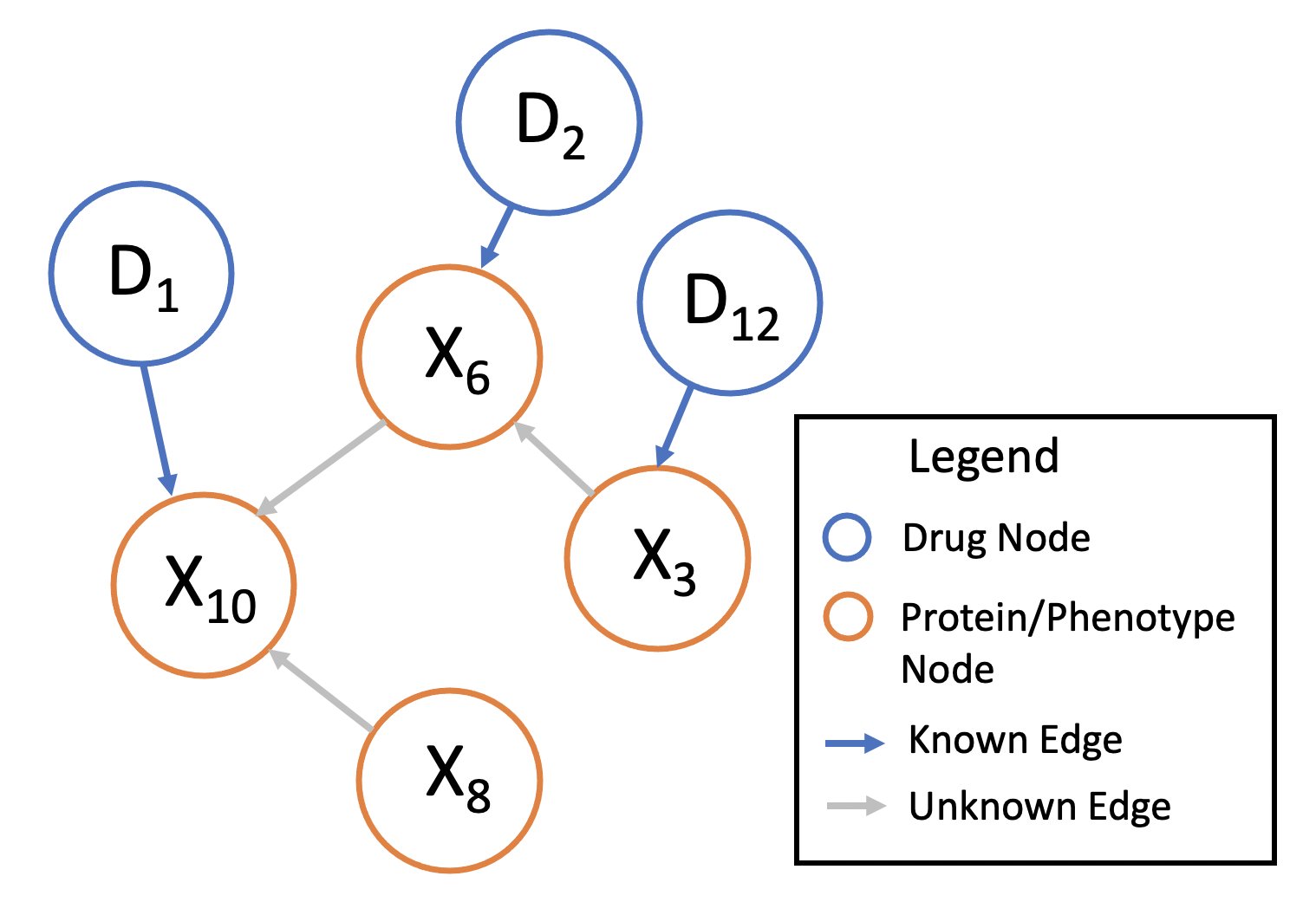}
    \caption{Causal Model}
    \label{fig:cm-dag}
   \end{subfigure}
\caption{a) Overview of perturbation data with a leave--one--drug--out (LODO) training / test set split. Drug 12 is never used in training. b) Causal graphical model for subset of drugs and responses (proteins and phenotypes). Drugs are exogenous variables with known targets, e.g. it is assumed known that drug 12 directly effects protein $X_3$.\label{fig:train-test}}
\end{figure}

% 
% \begin{center}
% \begin{figure}
% \centering
% \includegraphics[width=0.45\linewidth]{figs/train-test.png}
% \includegraphics[width=0.45\linewidth]{figs/drug-prot-dag.png}
% \caption{Left: Overview of perturbation data with a leave--one--drug--out (LODO) training / test set split. Drug 12 is never used in training. Right: Causal graphical model for subset of drugs and responses (proteins and phenotypes). Drugs are exogenous variables with known targets, e.g. it is assumed known that drug 12 directly effects protein $X_3$.\label{fig:train-test}}
% \end{figure}
% \end{center}

In perturbation prediction, drug doses (blue columns) are used to predict protein and phenotype responses (orange columns). Data is divided into training and test sets. Using only the training data, a model is constructed which can predict protein/phenotypes from the drug concentrations. The performance of the model is then evaluated using the test set drug concentrations and protein/phenotypes. When the test set is a simple random sample of all perturbations, this setup matches the standard approach to fitting and validating predictive models. Models such as linear regression of response variables on the drug concentrations can be used.

In practice one would like to construct a model which can accurately predict the effect of untested drugs and in doing so identify perturbations with interesting responses for further follow up. Random fold cross validation does not represent this use case well because all drugs are used in training. A more challenging form of validation, leave--one--drug--out (LODO), more closely aligns with the intended scientific uses of the perturbation prediction model. Figure \ref{fig:train-test} left panel depicts a LODO training-test set split. Here drug 12 is left out of training set i.e. the concentration of drug 12 in the training data is always $0$ because drug 12 was never used in training perturbations. The test perturbations all use drug 12 so column 12 of the drug matrix is never $0$ in test. LODO prediction is challenging for regression models because there is no way for the model to determine the effect of drug 12 on the response variables. For example, coefficients in a linear regression of response on drugs will not be defined because the gram matrix is not invertible.

The direct targets of drugs are often known a priori. For example, a MEK inhibitor drug should directly reduce the expression of the MEK protein. Other changes in the system, could then be assumed to be a downstream effect of MEK inhibition. Using this information about drug targets, a causal model, which infers causal relations among the protein and phenotype response variables, can be used to predict responses in LODO validation. The approach is graphically summarized in Figure \ref{fig:train-test} right. For clarity only a small number of the drug and response variables are shown. Drugs (blue nodes) are known to target (blue arrows) particular proteins (orange nodes). For example drug $D_{12}$ targets protein $X_3$. The causal relations among proteins is unknown a priori (grey arrows). Training set perturbations can be used to identify and estimate causal effects among the proteins. Then the effect of an untested perturbation, e.g. drug 12, can be determined by first assuming that the direct effect of drug 12 will be on protein 3, and then propagating this effect through the inferred protein network.

\cite{yuan2021cellbox} developed an ordinary differential equation (ODE) based model termed Cellbox and tested it on the Melanoma cell line data, both using random fold and LODO forms of validation. Cellbox outperformed all competing algorithms on both tasks. In the following section, we derive analytic results relating Cellbox to the simpler approach of regression of response variables on drugs.

\hypertarget{cellbox}{%
\section{Causal versus Regression Models for Prediction}\label{cellbox}}

\subsection{Linear Regression Model}
\label{subsec:reg}

We consider a multivariate linear model for drug response prediction and the challenge this model faces with LODO validation. Let $\bs{X} \in \mathbb{R}^{n \times p}$ be the matrix of protein and phenotype responses (orange matrix from Figure \ref{fig:lodo-struct}) and $\bs{D} \in \mathbb{R}^{n \times q}$ be the matrix of drug concentrations (blue matrix from Figure \ref{fig:lodo-struct}). Consider a model
\begin{equation*}
\bs{X} = \bs{D}R + \delta
\end{equation*}
with $R \in \mathbb{R}^{q \times p}$ where $R_{ij}$ is the effect of drug $i$ on response $j$. The rows of $\delta \in \mathbb{R}^{n \times p}$ are independent and $\mathbb{E}[\delta] = 0$. Consider estimating $R$ with
\begin{equation}
\label{eq:regress}
\widehat{R} = \argmin{R} ||\bs{X} - \bs{D}R||_F^2 + \lambda ||R||_1
\end{equation}
where $||\cdot||_F$ indicates the Frobenius norm and $||R||_1 = \sum_{i,j} |R_{i,j}|$. The sparsity inducing penalty term $\lambda ||R||_1$ may be used to force each drug to effect only a small number of response variables. For a new condition with drug concentrations $d \in \mathbb{R}^q$ applied, the predicted response is $\widehat{x} = d^T\widehat{R}$.

For the Melanoma cell line data $q=12$ and $n=89$, so regularization will not generally be necessary to obtain a unique solution from Equation \ref{eq:regress}. However with LODO validation, one column of $\bs{D}$ will be identically $0$, implying that without regularization ($\lambda=0$), the objective function in Equation \eqref{eq:regress} will not have a unique solution. Specifically if drug $j$ is held out, then $||\bs{X} - \bs{D}R||_F^2$ will be insensitive to changes in $R_{j,\cdot}$, the $j^{th}$ row of $R$. A solution to this problem is to use some amount of regularization. If $\lambda > 0$, then $\widehat{R}_{j,\cdot} = 0$ because this will minimize the penalty portion of the objective function. While this produces a unique solution, it seemingly will not produce a good estimator because it is unlikely that drug $j$ actually has no effect on the response variables.

\subsection{Causal Model}
\label{sec:causal-model}

We now discuss Cellbox, a causal modeling strategy which can overcome some of the limitations of regression with LODO validation. Cellbox was introduced in \cite{yuan2021cellbox}. First we summarize the Cellbox modeling and fitting procedure, modifying notation in certain instances for clarity.\footnote{See Model Configuration section of METHOD DETAILS in Star Methods of \cite{yuan2021cellbox} for original exposition of model.}

Cellbox uses a system of ordinary differential equations (ODEs) to model how proteins and phenotypes influence each other across time. Let $x^k(t,\theta) \in \mathbb{R}^p$ be the log--normalized change at time $t$ (relative to time $0$) of a set of $p$ proteins and phenotypes under perturbation condition $k$. The unknown parameters $\theta$ control how proteins influence each other. In condition $k$, drug concentrations $d^k \in \mathbb{R}^q$ are applied. Let $B \in \mathbb{R}^{p \times q}$ be the direct effects of a drug on the system with $B_{ij}$ the effect of 1 unit of drug $j$ on protein $i$. Define $u^k = Bd^k$, the direct effect of applying drug concentrations $d^k$ to the system. Since $d^k$ and $B$ are assumed known, $u^k$ is known as well. % For example suppose $d^k = (0,1.5,0,\ldots,0)^T$, then drug 2 is applied at a concentration of $1.5$. If the second column of $B$ is $(0,0,2,0,0,\ldots)^T$, then $u^k = (0,0,3,0,\ldots)^T$.

Response $i$ (protein or phenotype) under perturbation $k$ is modeled by
\begin{equation}
\label{eq:cellbox}
\frac{\partial x^{k}_i(t,\theta)}{\partial t} = \epsilon_i \phi \left( \sum_{j \neq i}w_{ij}x_j^k(t,\theta) - u_i^{k}\right) - w_{ii} x_i^{k}(t,\theta).
\end{equation}
The parameter $\theta = (W,\epsilon)$ where $w_{ij}$ for $j \neq i$ represents the causal effect of $x_j$ on $x_i$, $w_{ii}$ characterizes the effect of decay (the tendency of protein $i$ to return to the original level before perturbation), and $\epsilon_i$ controls the saturation effect of the protein. Cellbox can be fit with several envelope functions $\phi$ including identity, clipped linear, and sigmoid.

In \cite{yuan2021cellbox} Cellbox was fit with response variables measured at a single time point 24 hours after perturbation initiation. It was assumed that by this time, the system has reached steady state. Define the observed changes of the proteins/phenotypes at 24 hours using $x^k_i$. The steady state (equilibrium) changes implied by the model is
\begin{equation*}
x^k_i(\theta) \equiv \lim_{t \rightarrow \infty} x^k_i(t,\theta).
\end{equation*}
To estimate parameters $\theta$, discrepancy between predicted responses ($x^k_i(\theta)$) and observed responses ($x^k_i$) was computed with a $L_1$ (lasso) penalty term to induce sparsity on the off--diagonal elements of $W$. Specifically, the Cellbox algorithm  seeks $\theta$ which minimizes
\begin{equation*}
L(\theta) = \sum_k \sum_i |x^k_i - x^k_i(\theta)|^2 + \lambda ||W - diag(W)||_1
\end{equation*}
where $diag(W) \in \mathbb{R}^{p \times p}$ is the diagonal component of $W$. Given a candidate $\theta$, an ODE solver can be used to approximate $x^k_i(\theta)$. Subsequently $\theta$ are updated using gradient descent with automatic differentiation to determine 
\begin{equation}
\label{eq:cellbox-opt}
\widehat{W},\widehat{\epsilon} = \argmin{\theta=(W,\epsilon)} L(\theta).
\end{equation}
%To predict the response $x^{k^*}$ for some combination of drugs $d^{k^*}$, first the direct effects of the drug are determined ($u^{k^*} = Bd^{k^*}$), followed by an ODE solver to approximate $x^{k^*}(\widehat{\theta})$.

We now show that using a linear envelope function $\phi$ and setting $\epsilon=1$, the parameters $\theta$ in Cellbox can be derived as the solution to a penalized regression model. This result is a consequence of the fact that linear ODEs have closed form solutions and do not require numerical integration. This result facilitates a comparison of Cellbox with models which regress responses on drugs directly.

\begin{Thm} \label{thm:cellbox-steady}
Suppose $\phi$ is a linear envelope function, $\epsilon=1$, and $W$ is invertible. Then
\begin{equation}
\label{eq:cell_model}
x^k(\theta) = (x^k_1(\theta),\ldots,x^k_p(\theta))^T =  -W^{-1}Bd^k
\end{equation}
and
\begin{equation}
\label{eq:cell_pen}
\widehat{W} = \argmin{W} ||\bs{X} - \bs{D}B^T(-W^{{-1}^T})||_F^2 + \lambda ||W - diag(W)||_1.
\end{equation}
\end{Thm}

Equation \ref{eq:cell_pen} shows that for a simplified version of Cellbox ($\phi$ identity and $\epsilon=1$), ODE solvers are not necessary for estimating parameters $\widehat{W}$ and making predictions. For a new combination of drugs $d \in \mathbb{R}^q$, Cellbox predicts a response $\widehat{x} = d^TB^T(-\widetilde{W}^{-1^T})$. Qualitatively, the direct effect of the drugs are first determined $d^TB^T$ and then these direct effects are propagated through the system by multiplying the direct effects by $-\widetilde{W}^{-1^T}$.

The relative merits of the Regression and Causal modeling strategies are summarized as follows:
\begin{itemize}
\item \textbf{p versus q}: The Causal model estimates $W$ which consists of $p^2$ parameters, corresponding to direct effects of all response variables on each other. The Regression model estimates $R$ which consists of $qp$ parameters, corresponding to the total effect of each of the $q$ drugs on the $p$ response variables. Thus when $q < p$, regression requires estimating fewer parameters. Further, the dimension of the column space of $\bs{D}B^T$ is bounded by $min(q,p)$. This implies that when $q < p$, unregularized estimates of $W$ are not possible. For the Melanoma data (where $p=87$ and $q=12$) the regression model requires estimation of far fewer parameters. The regression model may be fit without regularization while the causal model will need regularization in order for the objective function to have a unique minimum.
\item \textbf{B Assumption}: The causal model requires knowledge of $B$, the direct effects of interventions on response variables. If $B$ is unknown or contains a large amount of uncertainty, then regression may be preferred.
\item \textbf{Generalization to New Drugs}: Consider the unregularized version of Estimator \eqref{eq:cell_pen}. The optimization problem will have a unique solution whenever $\bs{D}B^T \in \mathbb{R}^{n \times p}$ is full column rank. This can occur even if a particular drug is never used (one column of $\bs{D}$ is identically $0$). This represents an advantage over regression approaches. In contrast, the Regression objective function does not have a unique minimum with untested drugs when not using regularization. With regularization, the effect estimates of the untested drug is $0$.
\item \textbf{Interpretability:} The causal model is more interpretable because it estimates matrix $W$ (which in the Melanoma application is a gene regulatory network) which encodes how response variables causally effect each other.
%\item \textbf{Non-Targeted Proteins}: Suppose protein $k$ is not targeted by any drug. This is equivalent to the kth column of $B^T$ being $0$. This implies that the kth column of $G = DB^T$ is $0$. Thus it will not be possible to estimate the k row of $-W^{-1^T}$. This seems to be a major problem with $q < p$ and the drugs do not target many proteins (i.e. number of downstream effectors is small). This is the case with Cellbox.
\end{itemize}

While the analytic results comparing causal modeling to regression are only applicable to the linear case, several of these qualitative conclusions apply more generally. For example, regardless of what form of Causal model is used (e.g. non--linear version of Cellbox), inferences will be sensitive to misspecification of $B$ (the direct targets of the drugs), while regression (possibly non--linear) is not. Regression models, even non--linear ones, will not be able to generalize well to predict the effect of drugs not used in the training data.

Finally we note that the implementation of Cellbox in \cite{yuan2021cellbox} set elements of $W$ to $0$ which represent phenotype to protein causal effects. This is accomplished by restricting the domain of the parameter optimization in Equation \eqref{eq:cellbox-opt}. This enforces the domain knowledge that proteins may influence phenotypes but not vice versa. For clarity of exposition, we do not impose the conditions here or in the simulations since they are not directly relevant for understanding the relationship between regression and causal predictive models. However in the application to the Melanoma cell line, Section \ref{mela}, we follow \cite{yuan2021cellbox} and impose the restriction.

\subsection{Static Time Causal Structural Equation Modeling Interpretation}

We now show that Model \eqref{eq:cell_model} has a static time interpretation in terms of a linear structural equation model and directed acyclic graphs (DAGs). Let $x^k(A) \in \mathbb{R}^p$ represent noise--free responses when drug $d^k \in \mathbb{R}^q$ is applied to the system. As before, the matrix $B \in \mathbb{R}^{p \times q}$ represents known causal effects of drugs on response variables. The matrix $A \in \mathbb{R}^{p \times p}$ encodes the causal structure of $x^k$ with element $A_{ij}$ representing the causal effect of $x_j$ on $x_i$. Assuming $A$ represents a directed acyclic graph, one may write
\begin{equation}
\label{eq:seq-interp}
x^k_i(A) \gets \sum_{j=1}^p A_{ij}x^k_j(A) + (Bd^k)_i.
\end{equation}
One can write this system of equation as
\begin{equation*}
x^k(A) = Ax^k(A) + Bd^k.
\end{equation*}
Since $A$ represents a DAG, the matrix $I-A$ is invertible, so we have
\begin{equation}
\label{eq:dag-cellbox}
x^k(A) = (I-A)^{-1}Bd^k.
\end{equation}
Note that Model \eqref{eq:dag-cellbox} is identical to Model \eqref{eq:cell_model} with $A=I + W^T$. Thus the simplified Cellbox model ($\phi$ is identity and $\epsilon=1$), contains the linear DAG model as a special case. The simplified Cellbox model is more general than the linear DAG because it can contain self loops and cycles. For simulations in Section \ref{sim} we use the $A$ parameterization defined here.

% 
% . When $A$ represents a DAG, the element $A_{ij}$ is the causal effect of $Y_j$ on $Y_i$. The error term $\epsilon_Y$ is not observed and correlation among its elements models hidden confounding.
% 
% The goal is to estimate the causal effect of setting the drug nodes $D$ to some particular value $d$. Following the definitions of \cite{pearl2009causal}, the empty set satisfies the backdoor criteria for the $(D,X)$ nodes. Thus we have
% \begin{equation*}
% \mathbb{E}[X|do(D=d)] = \mathbb{E}[X|D=d].
% \end{equation*}
% Note that the conditional distribution of $X$ given $D$ may be represented using
% \begin{equation*}
% X = (I-A)^{-1}BD + (I-A)^{-1}\delta_Y + \delta_X.
% \end{equation*}
% There are two approaches to estimating the conditional expectation:
% \begin{itemize}
% \item \textbf{Regression Model:} Regress $X$ on the drug concentrations $D$, possibly with a penalty term to induce sparsity. This results in Estimator \eqref{eq:regress}.
% \item \textbf{Causal Model:} Regress $X$ on $BD$ to estimate $A$. With a sparsity inducing penalty $||A-diag(A)||$ we have
% \begin{equation*}
% \widehat{A} = \argmin{A} ||\bs{X} - \bs{D}B^T(I-A)^{-1}||_F^2 + \lambda ||A - diag(A)||_1^2.
% \end{equation*}
% Here $\widehat{A} = I + \widetilde{W}^T$.
% \end{itemize}
% 
%Note that the forms of the regularization penalties could change here without effecting conclusions. For example, if Cellbox had penalized diagonal terms of $W$, i.e. used penalty term $||W||_1^2$ instead of $||W - diag(W)||_1^2$, there would still be an equi

\hypertarget{sim}{%
\section{Simulation}\label{sim}}

We conduct a simulation to compare the relative performance of the regression estimator defined in Equation \eqref{eq:regress} and the causal estimator defined in Equation \eqref{eq:cell_pen}. We consider a problem with $p=5$ response variables and a total of $q=15$ drugs. Since $q > p$ (number of drugs is greater than number of response variables), regularization is not necessary ($\lambda$ is set to $0$ in all the simulations).

\begin{figure}
\centering
\includegraphics[width=0.9\linewidth]{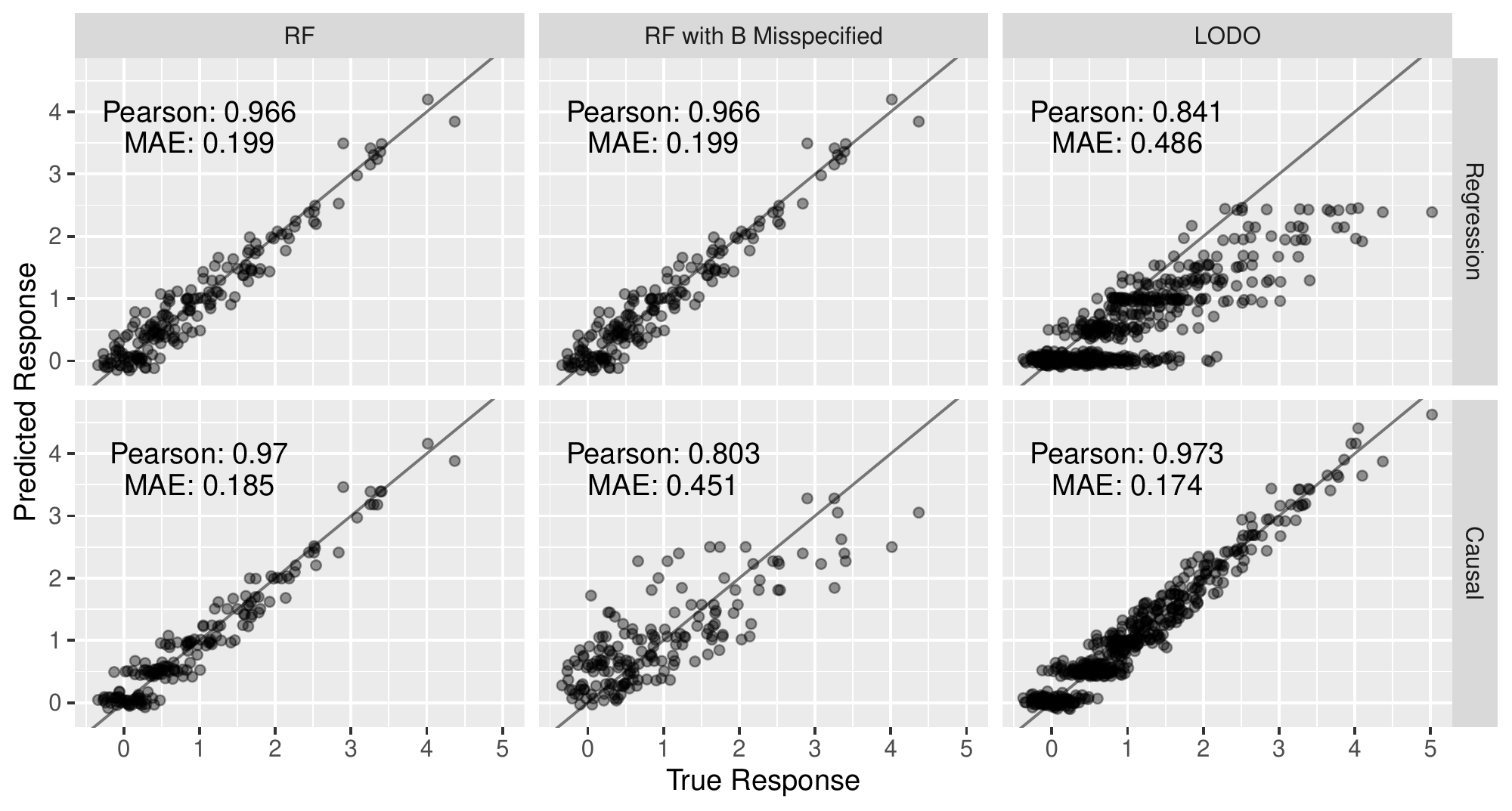}
\caption{Comparison of performance of regression model to causal model in various settings. The x-axis are the true responses and the y-axis is the predicted response. Regression and causal modeling perform similarly for random fold (RF) cross validation. For RF cross validation with $B$ misspecified, the causal model, which uses $B$, performs poorly. For leave--one--drug--out cross validation, the regression model performs poorly because it cannot model the effect of the left out drug on the responses.\label{fig:sim1-4}}
\end{figure}

Five drugs target a single response variable and $10$ drugs target two of the response variables. Drugs with a single target have a strength of $1$ while drugs with $2$ targets have a strength of $0.5$ for each target. All possible combinations of 2 drugs are applied to the system, thus there are a total of $n= {15 \choose 2} = 105$ observations.

The variable $X_1$ has a causal effect of $1.6$ on $X_2$ and $1.2$ on $X_3$. The variable $X_3$ has a causal effect of $2$ on $X_4$. All other causal effects among the response variables are $0$. The structure of $\bs{D}$, $B$, and $A$ is given in Equations \eqref{eq:sim-eq}.

\begin{equation}
\label{eq:sim-eq}
\bs{D} = 
\begin{pmatrix}
1 & 1 & 0 & 0 & \cdots & 0 & 0\\
1 & 0 & 1 & 0 & \cdots & 0 & 0\\
\vdots & \vdots & \vdots & \vdots & \vdots & \vdots & \vdots \\
1 & 0 & 0 & 0 & \cdots & 0 & 1\\
0 & 1 & 1 & 0 & \cdots & 0 & 0\\
0 & 1 & 0 & 1 & \cdots & 0 & 0\\
\vdots & \vdots & \vdots & \vdots & \vdots & \vdots & \vdots \\
0 & 1 & 0 & 0 & \dots & 0 & 1\\
\vdots & \vdots & \vdots & \vdots & \vdots & \vdots & \vdots \\
0 & 0 & 0 & 0 & \dots & 1 & 1\\
\end{pmatrix} \in \mathbb{R}^{105 \times 15} \, \, \, \,
B^T = 
\begin{pmatrix}
1 & 0 & 0 & 0 & 0 \\
0 & 1 & 0 & 0 & 0 \\
0 & 0 & 1 & 0 & 0 \\
0 & 0 & 0 & 1 & 0 \\
0 & 0 & 0 & 0 & 1 \\
\frac{1}{2} & \frac{1}{2} & 0 & 0 & 0\\
\frac{1}{2} & 0 & \frac{1}{2} & 0 & 0\\
\vdots & \vdots & \vdots & \vdots & \vdots \\
0 & 0 & 0 & \frac{1}{2} & \frac{1}{2}
\end{pmatrix} \in \mathbb{R}^{15 \times 5}  \, \, \, \, 
% 
% \begin{pmatrix}
% 1 & 0 & 0 & 0 & 0 & 1/2 & 1/2 & 0 &\cdots &0\\
% 0 & 1 & 0 & 0 & 0 & 1/2 & 0 & 0 &\cdots &0\\
% 0 & 0 & 1 & 0 & 0 & 0 & 1/2 & 0 &\cdots &0\\
% 0 & 0 & 0 & 1 & 0 & 0 & 0 & 0 &\cdots &1/2\\
% 0 & 0 & 0 & 0 & 1 & 0 & 0 & 0 &\cdots &1/2
% \end{pmatrix} \in \mathbb{R}^{5 \times 15}  \, \, \, \, \, \, \, \,
A = \begin{pmatrix}
0 & 0 & 0 & 0 & 0\\
1.6 & 0 & 0 & 0 & 0\\
1.2 & 0 & 0 & 0 & 0\\
0 & 0 & 2 & 0 & 0\\
0 & 0 & 0 & 0 & 0
\end{pmatrix}
\end{equation}

We simulate responses using the structural equation model in Equation \eqref{eq:seq-interp} with

\begin{equation*}
\bs{X} = \bs{D}B^T\left(I-A\right)^{-1} + \delta_X
\end{equation*}
where $\delta_X \in \mathbb{R}^{105 \times 5}$ with all elements independent distributed $N(0,0.2^2)$. We fit the regression and causal estimators under three settings:
\begin{itemize}
\item \textbf{Random Fold (RF):} The data is divided randomly into 2/3 training and 1/3 test. Since the training--test set split is random, every drug is used in training.
\item \textbf{RF with B Misspecified:} The training--test set split is identical to RF. However the $B$ matrix (direct effect of drugs) is misspecified. Instead of using the correct $B$, the 10 drugs with 2 targets are assumed to influence their targets with a strength of $1$ (rather than the correct value of $0.5$).
\item \textbf{Leave-one-drug-out (LODO):} One drug is left out of the training set and used as test. For the regression estimator, the coefficient on the left out drug is set to $0$.
\end{itemize}

Results are summarized in Figure \ref{fig:sim1-4}. The true response values are plotted on the x--axis and the predicted response values are plotted on the y--axis. High correlations imply that the estimator is performing well in the respective setting. For Random Fold (RF) cross validation, both the Regression and Causal estimators perform well. In the RF with B Misspecified setting, Regression performs well and in fact makes identical predictions to the RF cross validation case because the Regression estimator does not depend on $B$. In contrast, the Causal estimator performs poorly because it uses an incorrectly specified $B$. Finally in the LODO setting, the Regression estimator performs poorly because it incorrectly infers that the effect of the held out drug on all response $0$. In contrast, the causal estimator performs well because it models the causal relations among the response variables which enables it to generalize predictions to untested drugs.

The Causal model produces an estimate of $A$. The true $A$ and the estimates $\widehat{A}$ for each validation setting are displayed graphically in Figure \ref{fig:sim1-4}. Note that LODO produces 15 estimates $\widehat{A}$ (one for each held out drug). We plot only one of them here. Edge widths are proportional to size of the coefficient estimate. For visual clarity, small effects (coefficients less than 0.2 in absolute size) are not displayed. The Random Fold $\widehat{A}$ in Figure \ref{fig:network_rf} and the LODO $\widehat{A}$ in Figure \ref{fig:network_LODO} closely resemble the true $A$ in Figure \ref{fig:network_true}. In contrast, Figure \ref{fig:network_RF_missB} shows that when $B$ is misspecified the resulting $\widehat{A}$ is a poor estimate.

\begin{figure}
 \centering
  \begin{subfigure}{0.4\textwidth}
   \centering
   \includegraphics[width=0.9\linewidth]{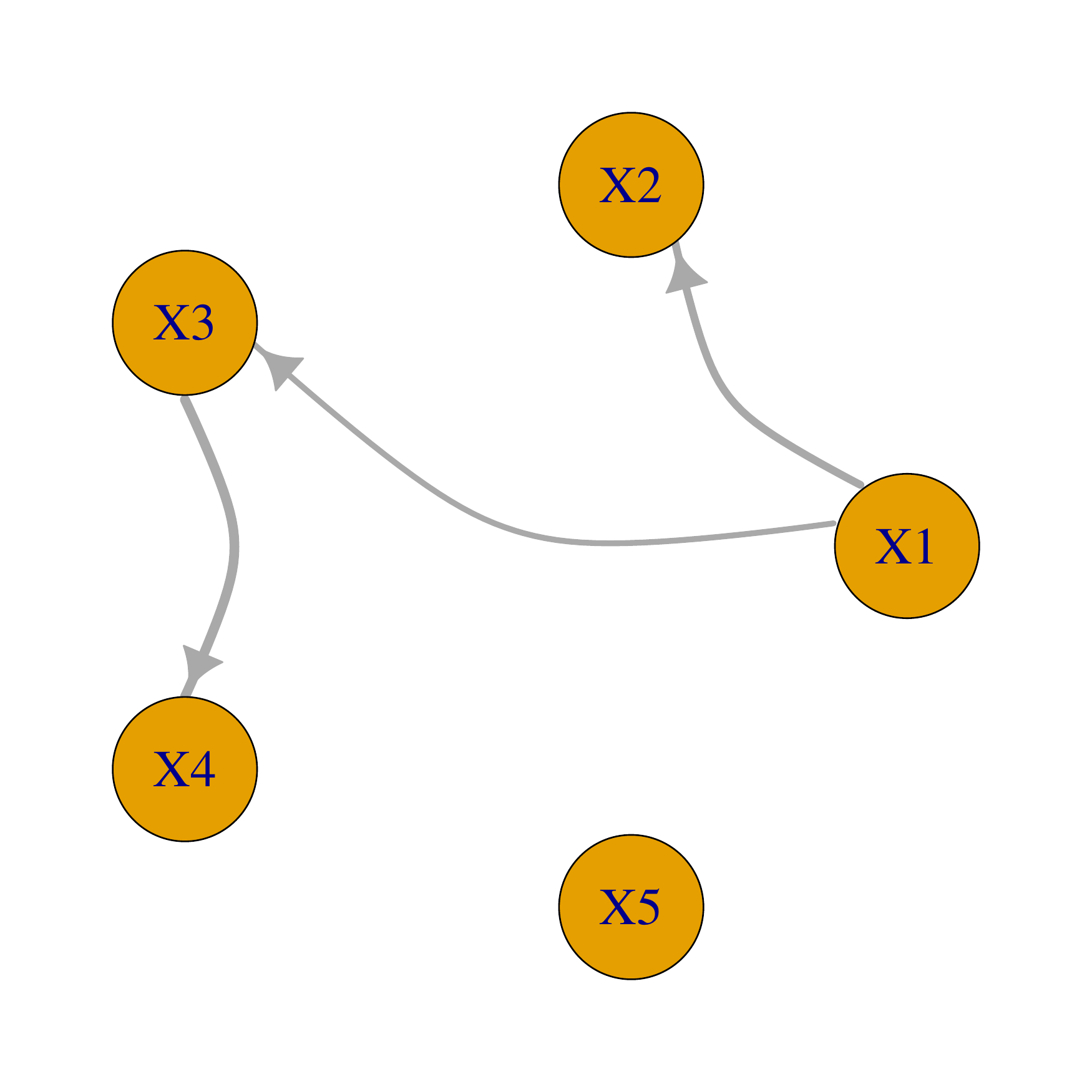}
   \caption{True Network}
   \label{fig:network_true}
   \end{subfigure}%
   \begin{subfigure}{0.4\textwidth}
    \centering
    \includegraphics[width=0.9\linewidth]{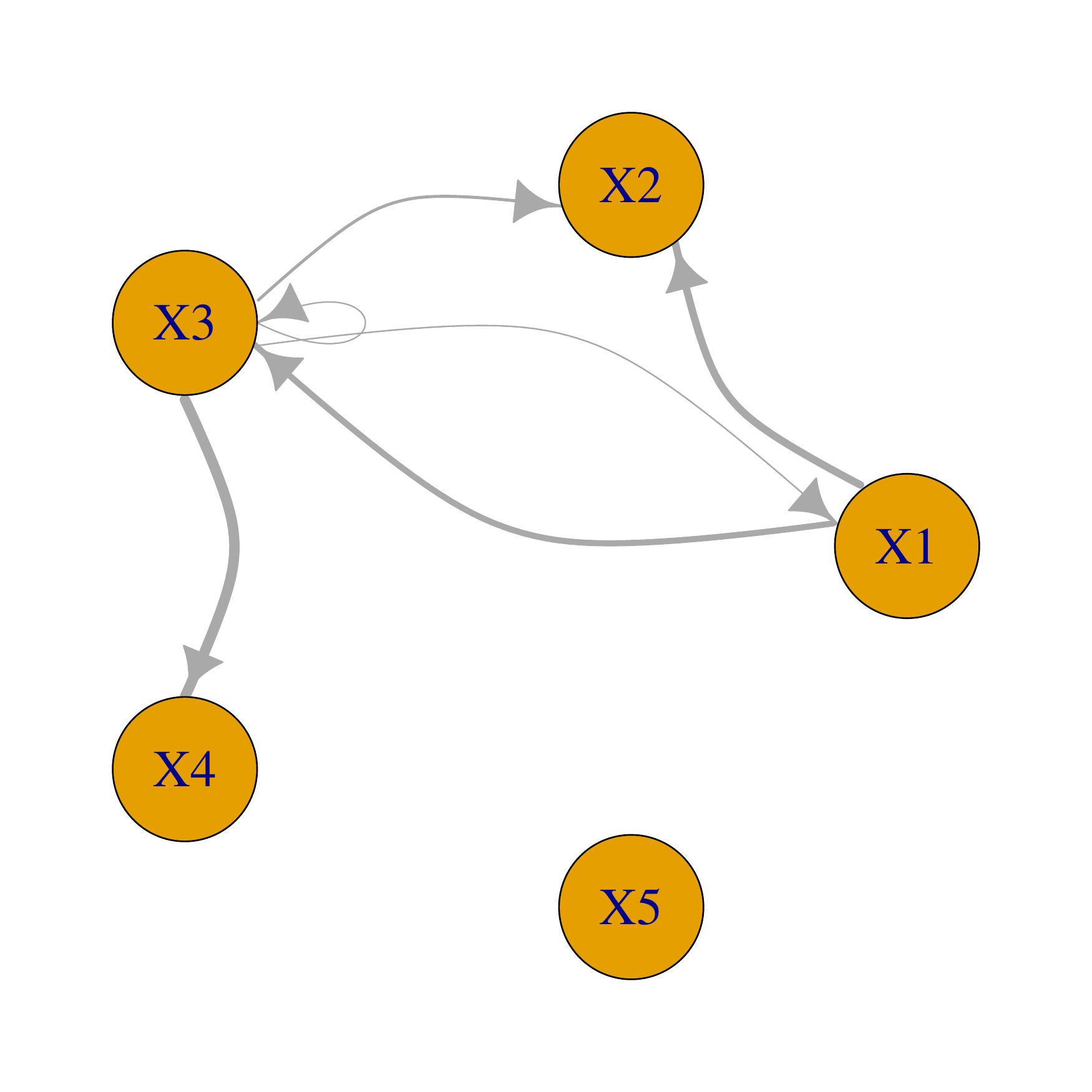}
    \caption{Random Fold}
    \label{fig:network_rf}
   \end{subfigure}
   \begin{subfigure}{0.4\textwidth}
    \centering
    \includegraphics[width=0.9\linewidth]{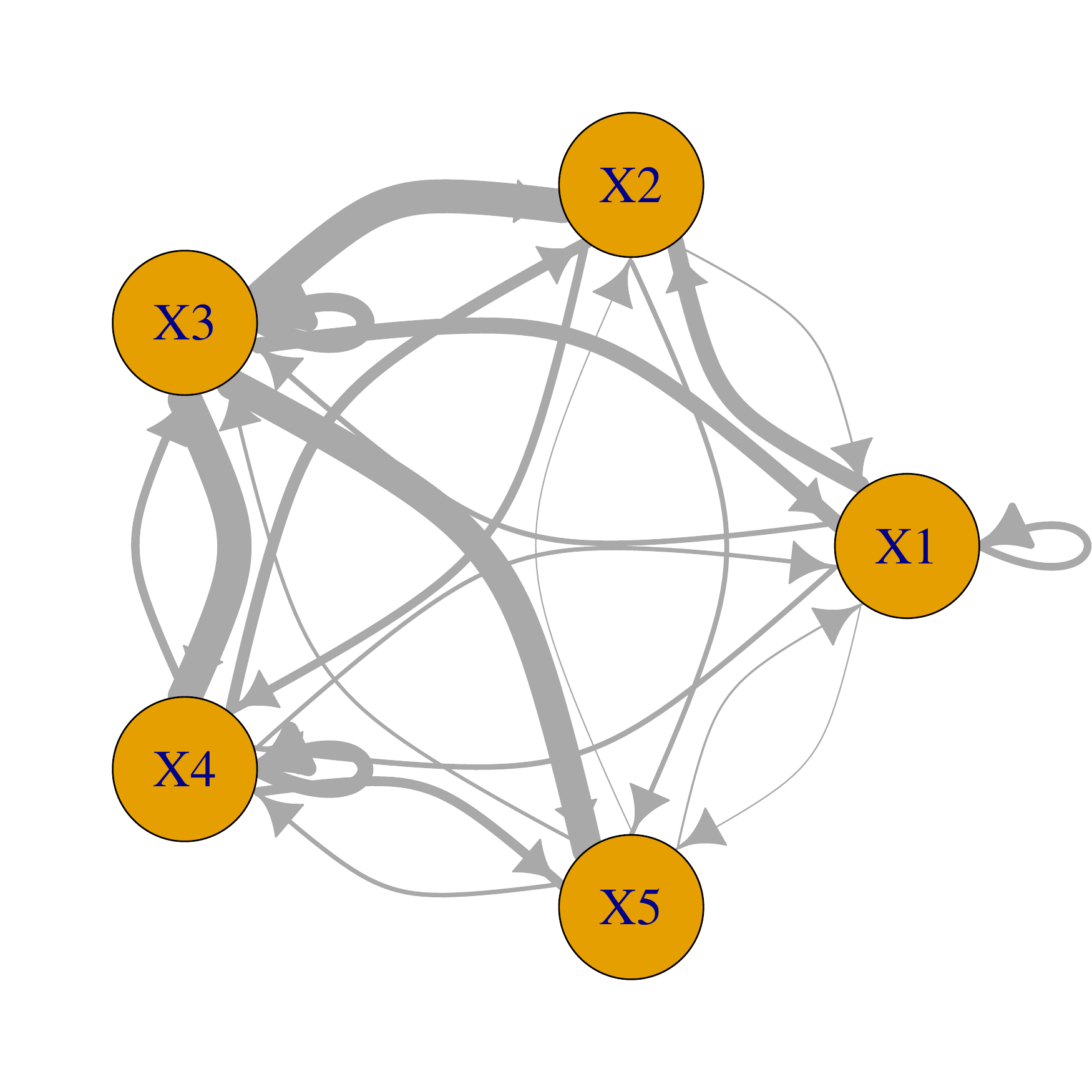}
    \caption{Random Fold, Misspecified B}
    \label{fig:network_RF_missB}
   \end{subfigure}
   \begin{subfigure}{0.4\textwidth}
    \centering
    \includegraphics[width=0.9\linewidth]{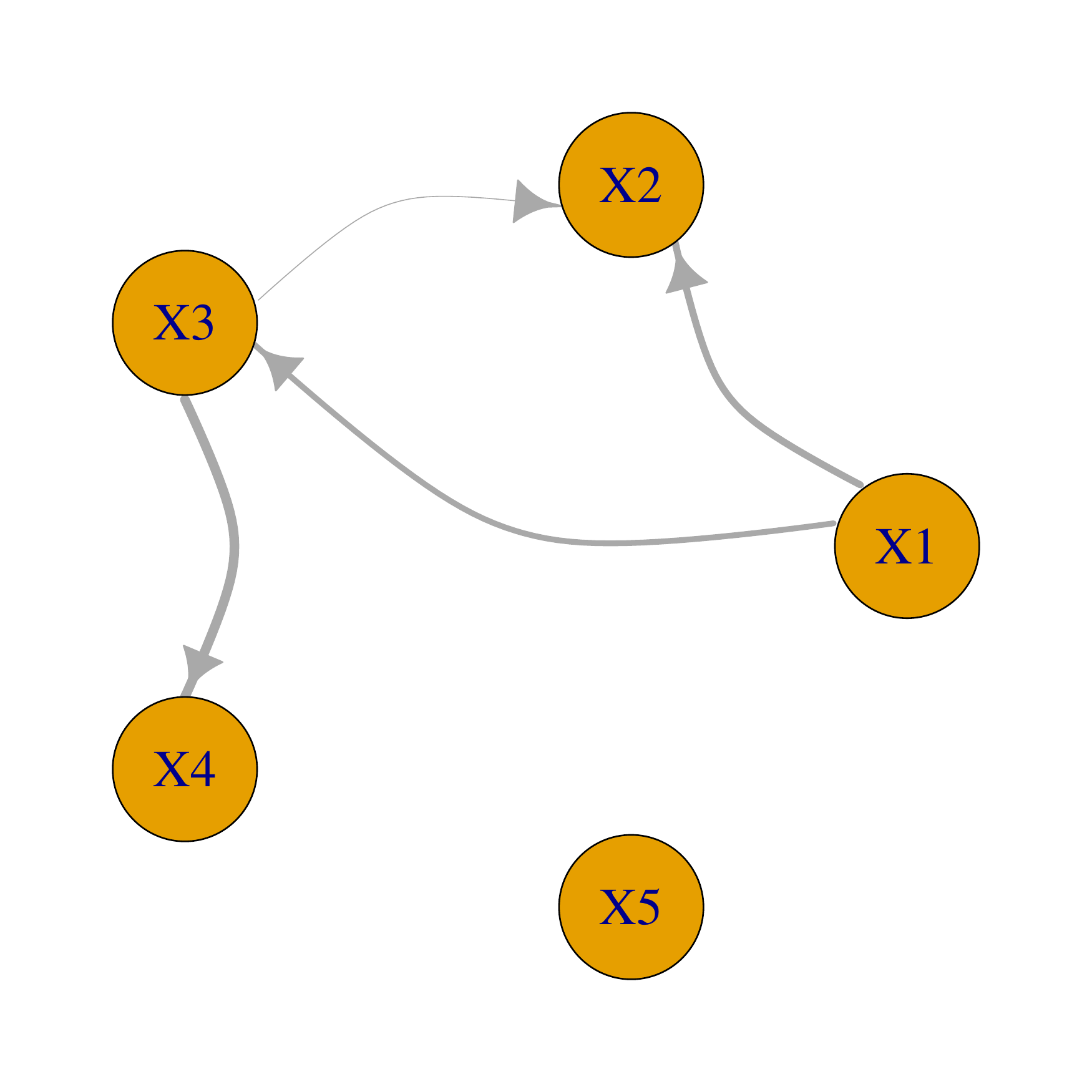}
    \caption{LODO}
    \label{fig:network_LODO}
   \end{subfigure}
\caption{True network and estimated networks for different simulation settings. For Random Fold and LODO, the estimated $A$ is quite close to the true $A$. For Random Fold with Misspecified $B$, the estimated $A$ contains many erroneous edges.}
\end{figure}

% 
% \begin{figure}
% \centering
% \includegraphics[width=0.7\linewidth]{figs/sim1-1}
% \caption{Random training test set split with correctly specified $B$.\label{fig:sim1-1}}
% \end{figure}
% 
% 
% \begin{figure}
% \centering
% \includegraphics[width=0.7\linewidth]{figs/sim1-2}
% \caption{Random training test set split with mis-specified $B$.\label{fig:sim1-2}}
% \end{figure}
% 
% 
% \begin{figure}
% \centering
% \includegraphics[width=0.7\linewidth]{figs/sim1-3}
% \caption{Leave one drug out traing test set split.\label{fig:sim1-3}}
% \end{figure}
% 
% 
% 

% see sim/sim1.R for simulation code

% 
% \section{Stability, Domain Adaptation, Transfer Learning}
% \label{domain-adaption}
% 
% \begin{itemize}
% \item 1--2 paragraph summarizing some of work by Buhlmann and others on the relationship between causality and stability, invariance, and domain adaptation.
% \item then discuss kemmeren data set a bit. perhaps include graphic describing structure with gene knockouts
% \item then discuss how this relates to extrapolation. essentially the goal is to predict gene expression given knockout information. if knockout has not been previously performed, we are extrapolating.
% \item propose new estimator
% \end{itemize}
% 

\hypertarget{mela}{%
\section{Melanoma Cell Line Perturbation Prediction}\label{mela}}

We compare Cellbox and Linear Regression for prediction of protein and phenotype responses in the Melanoma data set introduced in Section \ref{overview}. The two validation procedures we describe below follow the procedures in \cite{yuan2021cellbox}. Cellbox is implemented with a sigmoid activation function and tuning parameter selection as described in \cite{yuan2021cellbox}.

\subsection{Random Fold Cross Validation}

The 89 experimental conditions are split into 70\% training (62 conditions) and 30\% testing (27 conditions). Models are fit on the training set and used to predict the responses on the test set. This process is repeated 1000 times and the predictions averaged across these runs. Predicted responses versus experiment responses are plotted for Cellbox in Figure \ref{fig:test1} and Linear Regression in Figure \ref{fig:sub2}. The predictions from the linear model show a stronger correlation with the response than Cellbox (Pearson's correlation of 0.947 versus 0.926) and lower mean absolute error (0.093 versus 0.105). As discussed in Section \ref{sec:causal-model}, random fold cross validation favors regression models (relative to leave one drug out) because the regression model estimates fewer parameters and does not require regularization.

\subsection{Leave One Drug Out}

We now consider the more challenging Leave One Drug Out (LODO) validation where a drug is held out of training. For example, if the drug aMEK is held out, the training data is all conditions with aMEK concentration equal to 0 and the test set is all conditions where aMEK has been applied, either as monotherapy or in combination with other drugs. Since there are 12 drugs, there are 12 training--test set pairs.

The unregularized regression optimization problem from Equation \eqref{eq:regress} ($\lambda=0$) does not have a unique minimizer because the training data design matrix $\bs{D}$ is not full column rank. This is because the column of $\bs{D}$ corresponding to the held out drug is the $0$ vector. A regularized fit of Equation \eqref{eq:regress} (i.e. $\lambda \neq 0$) will result in a coefficient estimate for the left out drug of $0$. Fitting this estimator would involve tuning parameter selection with a method such as cross--validation. We choose a simpler approach of setting the coefficient for the left out drug to $0$ and fitting the unregularized estimator to the remaining columns of $\bs{D}$. This approach does not require tuning parameter selection. The effect is to assume that the drug held out has no effect on any of the response variables. This represents a crude benchmark model rather than an empirically motivated model assumption.

For each test set, we compute the correlation between the true responses and the predicted responses. This results in 12 correlations for Cellbox and Linear Regression. Figure \ref{fig:lodo} displays these correlations. Linear Regression slightly outperforms Cellbox with average correlation coefficient of $0.784$ as compared with $0.780$ for Cellbox.

\begin{figure}
 \centering
  \begin{subfigure}{0.33\textwidth}
   \centering
   \includegraphics[width=0.9\linewidth]{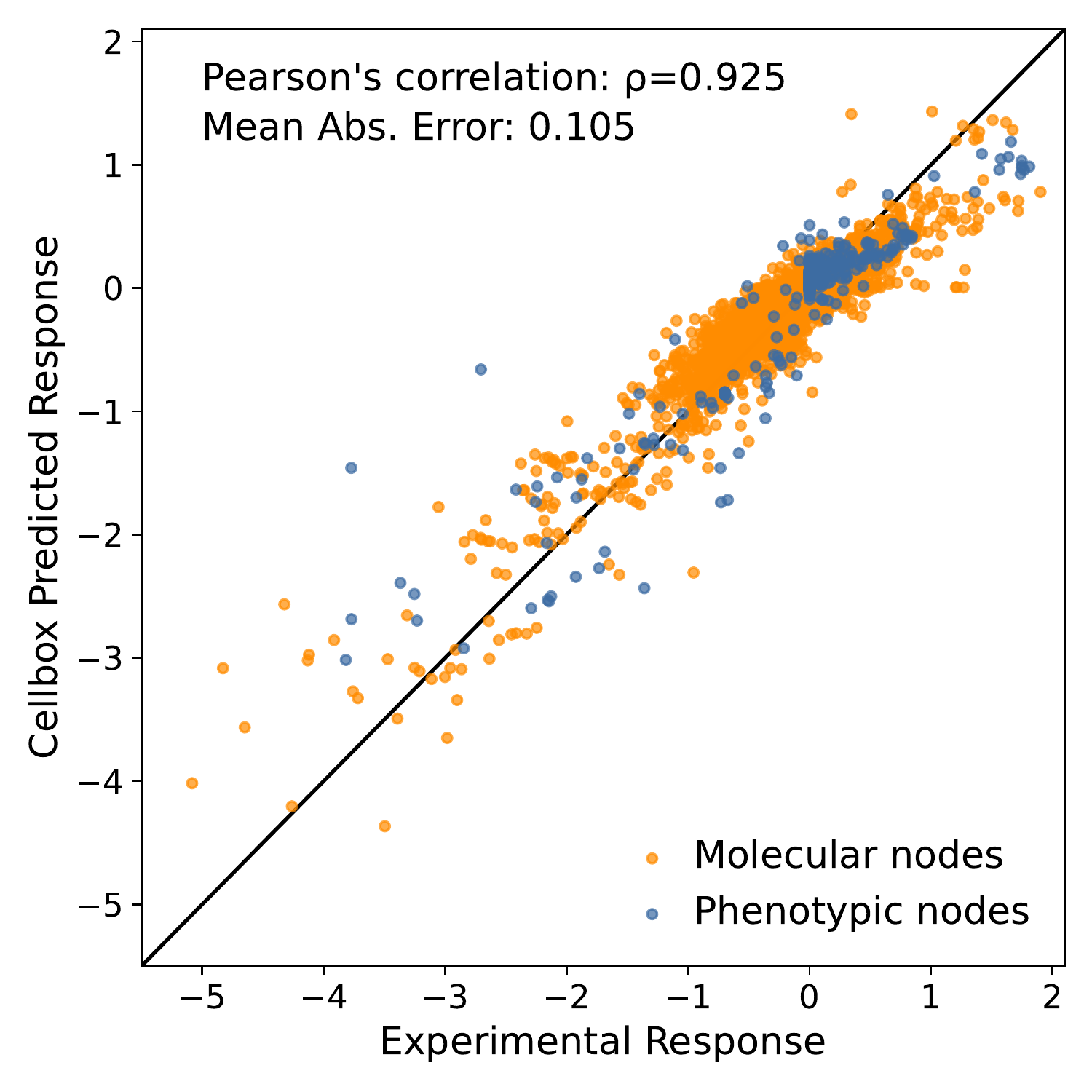}
   \caption{CellBox with RF}
   \label{fig:test1}
   \end{subfigure}%
   \begin{subfigure}{0.33\textwidth}
    \centering
    \includegraphics[width=0.9\linewidth]{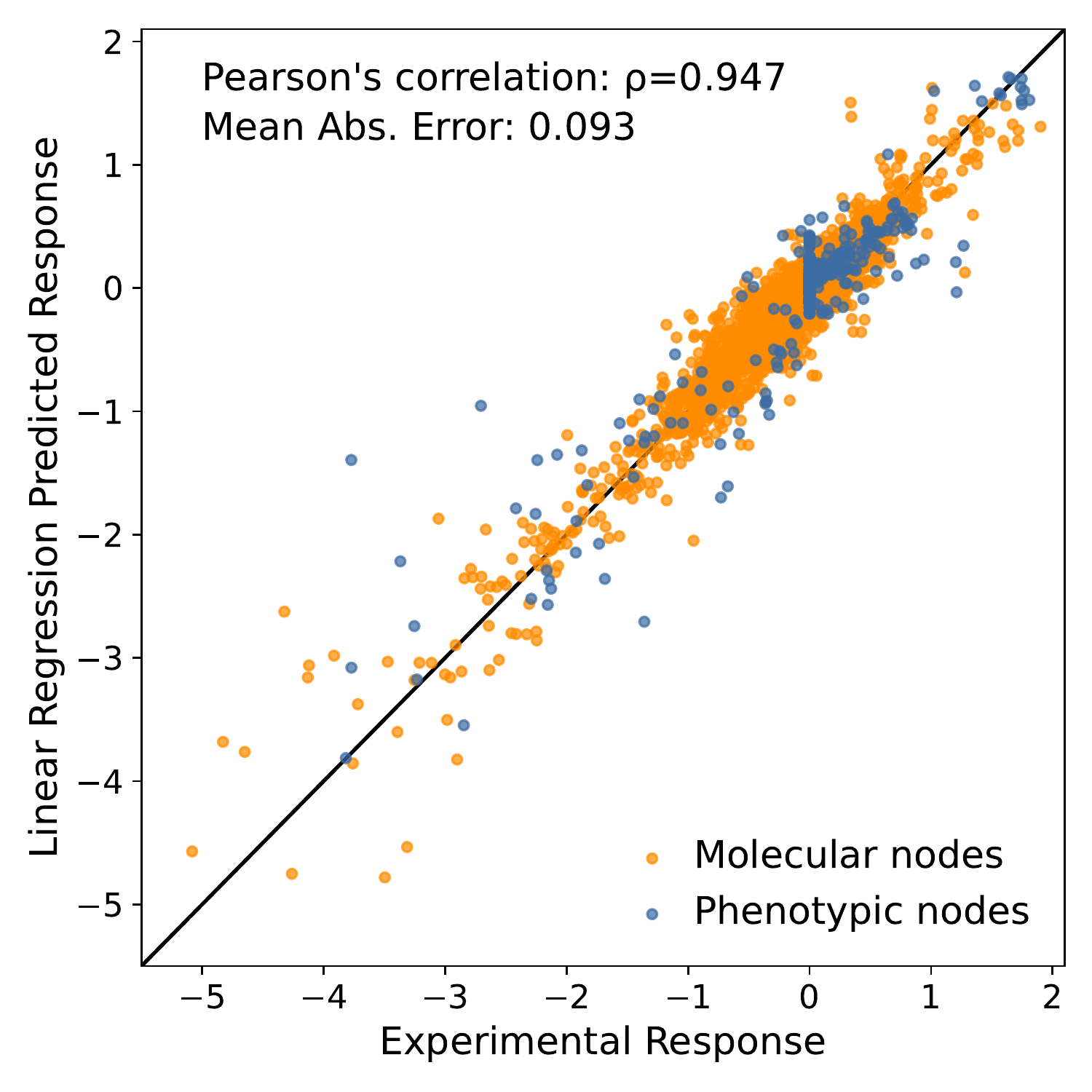}
    \caption{Linear Regression with RF}
    \label{fig:sub2}
   \end{subfigure}
   \begin{subfigure}{0.33\textwidth}
    \centering
    \includegraphics[width=0.9\linewidth]{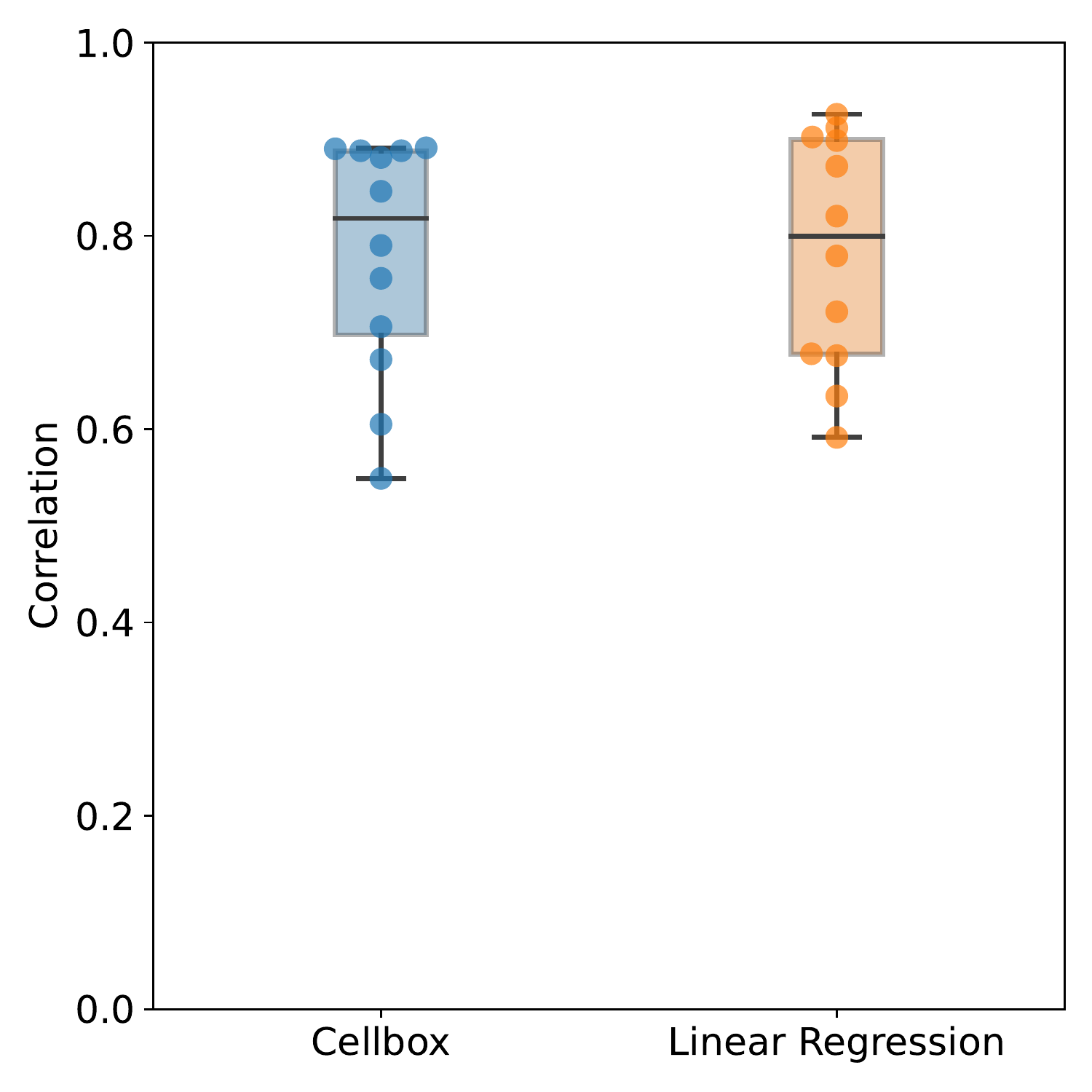}
    \caption{LODO}
    \label{fig:lodo}
   \end{subfigure}
\caption{Cellbox and Linear Regression performance on Random Fold (RF) and Leave One Drug Out (LODO) validation. Linear Regression has better performance in both settings.}
\end{figure}

% \begin{figure}
%  \centering
%  \includegraphics[width=6cm, height = 6cm]{figs/LODO}
%  \caption{Comparison of Linear Regression and Cellbox on LODO}
%  \label{fig:my_label}
% \end{figure}
% 
% 
% 

\hypertarget{discussion}{%
\section{Discussion}\label{discussion}}

The field of causal inference has historically focused on parameter estimation and hypothesis testing. Recently, several works have explored using causal models for prediction \citep{versteeg2019boosting,meinshausen2016methods,yuan2021cellbox,long2022sample}. Prediction performance is highly relevant for several scientific applications including cell--line perturbation response prediction. For causal prediction models to provide meaningful scientific insight, it is critical to understand their relationship with non--causal regression approaches and appropriately benchmark models when assessing performance.

Here we derived the first analytic results comparing a recently proposed causal prediction model, Cellbox, with regression models. The analytic results and simulations facilitated an improved understanding of the strengths and weaknesses of the two approaches to prediction. In brief, regression models are simpler to fit but lack an ability to extrapolate to new data settings, such as prediction of response to a drug not used in the training set.

Cellbox obtained state of the art performance on a Melanoma cell line perturbation data set, outperforming a Belief Propagation algorithm, a deep learning Neural Network (NN), and a co--expression model. Here we demonstrated that Cellbox, and hence all the competitor methods, failed to outperform linear regression in either random fold or leave--one--drug--out (LODO) cross validation. The latter finding is particularly surprising because this is a setting which favors causal modeling approaches. These results highlight that simple modeling strategies can be the most effective and are critically important when benchmarking performance of new models.

The Melanoma perturbation data set used here is relatively small. Larger perturbation experiments test hundred or thousand of perturbations across dozens of cell lines with responses measured at multiple time points \citep{subramanian2017next,zhao2020large}. These data sets are likely to be more favorable to causal modeling strategies because they may contain sufficient information to identify and constrain model parameters.

\hypertarget{proofs}{%
\section{Proofs}\label{proofs}}

\hypertarget{cellbox-steady}{%
\subsection{Proof of Theorem \ref{thm:cellbox-steady}}\label{cellbox-steady}}

With the identity envelope $\phi(\cdot) = \cdot$ and $\epsilon_i = 1$, Equation \eqref{eq:cellbox} simplifies to
\begin{equation}
\label{eq:cellbox-linear}
\frac{\partial x^{k}(t,\theta)}{\partial t} = Wx^k(t) - u^{k} = Wx^k(t) - Bd^k.
\end{equation}
Now rewrite the ODE to explicitly include drug (condition) nodes. Let $d^k(t) \in \mathbb{R}^q$ indicate the concentrations of the $q$ drugs at time $t$. Since drug concentrations are assumed constant $d^k(t) = d^k(0) = d^k$. Define
\begin{equation*}
y^k(t,\theta) = \begin{pmatrix}
d^k(t) \\
x^k(t,\theta)
\end{pmatrix} \in \mathbb{R}^{q + p}
\end{equation*}
where
\begin{equation}
\label{eq:cell_drug}
\frac{\partial y^k(t,\theta)}{\partial t} = \underbrace{\begin{pmatrix}
0 & 0\\
B & W
\end{pmatrix}}_{\equiv A}
y^k(t) = \begin{pmatrix}
0\\
Wx^k(t) + Bd^k
\end{pmatrix}.
\end{equation}
There is a simple closed form solution for the system at time $t$, specifically
\begin{equation*}
y^k(t,\theta) = e^{At}y^k(0,\theta).
\end{equation*}
See \cite{adkins2015ordinary} (Section 9.5 Theorem 2) for a derivation of this result. Further by Lemma \ref{cellbox-lemma}
\begin{equation*}
y^k(\theta) \equiv \lim_{t\rightarrow \infty} y^k(t,\theta) =  \begin{pmatrix}
d(0)^k \\
-W^{-1}Bd(0)^k
\end{pmatrix}.
\end{equation*}

Now we have
\begin{align*}
\widehat{W} &= \argmin{W} \sum_k \sum_i |x^k_i - x^k_i(\theta)|^2 + \lambda ||W - diag(W)||_1\\
&= \argmin{W} \sum_k ||x^k - x^k(\theta)||^2_2 + \lambda ||W - diag(W)||_1\\
&= \argmin{W} ||\bs{X} - \bs{D}B^T(-W^{-1^T})||_F^2 + \lambda ||W - diag(W)||_1
\end{align*}

\hypertarget{lemmas}{%
\subsection{Proof of Lemmas}}

\begin{Lem} \label{cellbox-lemma}
\begin{equation*}
e^{At} = \begin{pmatrix}
I & 0\\
-W^{-1}B & 0
\end{pmatrix}.
\end{equation*}
\end{Lem}

\begin{proof}

Note that
\begin{align*}
e^{At} &= I + \sum_{i=1}^\infty \frac{A^it^i}{i!}\\
&= I + \sum_{i=1}^\infty \frac{\begin{pmatrix}
0 & 0\\
W^{i-1}B & W^i
\end{pmatrix}t^i}{i!}\\
&= \begin{pmatrix}
I & 0\\
\sum_{i=1}^{\infty} \frac{W^{i-1}Bt^i}{i!} & \sum_{i=0}^\infty \frac{W^it^i}{i!}
\end{pmatrix}\\
&= \begin{pmatrix}
I & 0\\
W^{-1}(e^{Wt}-I)B & e^{Wt}
\end{pmatrix}.
\end{align*}

We now show that $\lim_{t \rightarrow \infty} e^{Wt} = 0$ which implies the desired result. Let $c_j$ for $j=1,\ldots,p$ be an eigenbasis for $W$ such that $W c_j = \lambda_j c_j$. Note that since $W \prec 0$, $\lambda_j < 0$ for all $j$. Consider any $r \in \mathbb{R}^p$ with basis decomposition $r = \sum_j \gamma_j c_j$. We have
\begin{align*}
e^{Wt}r &= \sum_{j=1}^p \left(\sum_{i=0}^\infty \frac{W^it^i}{i!}\right) \gamma_j c_j\\
&= \sum_{j=1}^p \left(\sum_{i=0}^\infty \frac{\lambda_j^it^i}{i!}\right)\gamma_j c_j\\
&= \sum_{j=1}^p e^{\lambda_j t}\gamma_j c_j\\
&\rightarrow 0.
\end{align*}
Since this is true for any $r$, $\lim_{t \rightarrow \infty} e^{Wt} \rightarrow 0$.
\end{proof}

\bibliographystyle{abbrvnat}
\bibliography{refs}

\begin{thebibliography}{13}
\providecommand{\natexlab}[1]{#1}
\providecommand{\url}[1]{\texttt{#1}}
\expandafter\ifx\csname urlstyle\endcsname\relax
  \providecommand{\doi}[1]{doi: #1}\else
  \providecommand{\doi}{doi: \begingroup \urlstyle{rm}\Url}\fi

\bibitem[Adkins and Davidson(2015)]{adkins2015ordinary}
W.~Adkins and M.~Davidson.
\newblock \emph{Ordinary Differential Equations}.
\newblock Undergraduate Texts in Mathematics. Springer New York, 2015.
\newblock ISBN 9781489987679.
\newblock URL \url{https://books.google.com/books?id=Lmb5sgEACAAJ}.

\bibitem[Korkut et~al.(2015)Korkut, Wang, Demir, Aksoy, Jing, Molinelli, Babur,
  Bemis, Sumer, Solit, et~al.]{korkut2015perturbation}
A.~Korkut, W.~Wang, E.~Demir, B.~A. Aksoy, X.~Jing, E.~J. Molinelli,
  {\"O}.~Babur, D.~L. Bemis, S.~O. Sumer, D.~B. Solit, et~al.
\newblock Perturbation biology nominates upstream--downstream drug combinations
  in raf inhibitor resistant melanoma cells.
\newblock \emph{Elife}, 4:\penalty0 e04640, 2015.

\bibitem[Long and Ha(2022)]{long2022sample}
J.~P. Long and M.~J. Ha.
\newblock Sample selection bias in evaluation of prediction performance of
  causal models.
\newblock \emph{Statistical Analysis and Data Mining: The ASA Data Science
  Journal}, 15\penalty0 (1):\penalty0 5--14, 2022.

\bibitem[Lotfollahi et~al.(2019)Lotfollahi, Wolf, and
  Theis]{lotfollahi2019scgen}
M.~Lotfollahi, F.~A. Wolf, and F.~J. Theis.
\newblock scgen predicts single-cell perturbation responses.
\newblock \emph{Nature methods}, 16\penalty0 (8):\penalty0 715--721, 2019.

\bibitem[Meinshausen et~al.(2016)Meinshausen, Hauser, Mooij, Peters, Versteeg,
  and B{\"u}hlmann]{meinshausen2016methods}
N.~Meinshausen, A.~Hauser, J.~M. Mooij, J.~Peters, P.~Versteeg, and
  P.~B{\"u}hlmann.
\newblock Methods for causal infåerence from gene perturbation experiments and
  validation.
\newblock \emph{Proceedings of the National Academy of Sciences}, 113\penalty0
  (27):\penalty0 7361--7368, 2016.

\bibitem[Rothenh{\"a}usler et~al.(2019)Rothenh{\"a}usler, B{\"u}hlmann,
  Meinshausen, et~al.]{rothenhausler2019}
D.~Rothenh{\"a}usler, P.~B{\"u}hlmann, N.~Meinshausen, et~al.
\newblock Causal dantzig: fast inference in linear structural equation models
  with hidden variables under additive interventions.
\newblock \emph{The Annals of Statistics}, 47\penalty0 (3):\penalty0
  1688--1722, 2019.

\bibitem[Sachs et~al.(2005)Sachs, Perez, Pe'er, Lauffenburger, and
  Nolan]{sachs2005causal}
K.~Sachs, O.~Perez, D.~Pe'er, D.~A. Lauffenburger, and G.~P. Nolan.
\newblock Causal protein-signaling networks derived from multiparameter
  single-cell data.
\newblock \emph{Science}, 308\penalty0 (5721):\penalty0 523--529, 2005.

\bibitem[Squires et~al.(2020)Squires, Shen, Agarwal, Shah, and
  Uhler]{squires2020causal}
C.~Squires, D.~Shen, A.~Agarwal, D.~Shah, and C.~Uhler.
\newblock Causal imputation via synthetic interventions.
\newblock \emph{arXiv preprint arXiv:2011.03127}, 2020.

\bibitem[Subramanian et~al.(2017)Subramanian, Narayan, Corsello, Peck, Natoli,
  Lu, Gould, Davis, Tubelli, Asiedu, et~al.]{subramanian2017next}
A.~Subramanian, R.~Narayan, S.~M. Corsello, D.~D. Peck, T.~E. Natoli, X.~Lu,
  J.~Gould, J.~F. Davis, A.~A. Tubelli, J.~K. Asiedu, et~al.
\newblock A next generation connectivity map: L1000 platform and the first
  1,000,000 profiles.
\newblock \emph{Cell}, 171\penalty0 (6):\penalty0 1437--1452, 2017.

\bibitem[Tibes et~al.(2006)Tibes, Qiu, Lu, Hennessy, Andreeff, Mills, and
  Kornblau]{tibes2006reverse}
R.~Tibes, Y.~Qiu, Y.~Lu, B.~Hennessy, M.~Andreeff, G.~B. Mills, and S.~M.
  Kornblau.
\newblock Reverse phase protein array: validation of a novel proteomic
  technology and utility for analysis of primary leukemia specimens and
  hematopoietic stem cells.
\newblock \emph{Molecular cancer therapeutics}, 5\penalty0 (10):\penalty0
  2512--2521, 2006.

\bibitem[Versteeg and Mooij(2019)]{versteeg2019boosting}
P.~Versteeg and J.~M. Mooij.
\newblock Boosting local causal discovery in high-dimensional expression data.
\newblock In \emph{2019 IEEE International Conference on Bioinformatics and
  Biomedicine (BIBM)}, pages 2599--2604. IEEE, 2019.

\bibitem[Yuan et~al.(2021)Yuan, Shen, Luna, Korkut, Marks, Ingraham, and
  Sander]{yuan2021cellbox}
B.~Yuan, C.~Shen, A.~Luna, A.~Korkut, D.~S. Marks, J.~Ingraham, and C.~Sander.
\newblock Cellbox: interpretable machine learning for perturbation biology with
  application to the design of cancer combination therapy.
\newblock \emph{Cell systems}, 12\penalty0 (2):\penalty0 128--140, 2021.

\bibitem[Zhao et~al.(2020)Zhao, Li, Chen, Luo, Ju, Nesser, Johnson-Camacho,
  Boniface, Lawrence, Pande, et~al.]{zhao2020large}
W.~Zhao, J.~Li, M.-J.~M. Chen, Y.~Luo, Z.~Ju, N.~K. Nesser, K.~Johnson-Camacho,
  C.~T. Boniface, Y.~Lawrence, N.~T. Pande, et~al.
\newblock Large-scale characterization of drug responses of clinically relevant
  proteins in cancer cell lines.
\newblock \emph{Cancer Cell}, 38\penalty0 (6):\penalty0 829--843, 2020.

\end{thebibliography}

\end{document}